\pdfoutput=1
\documentclass[letterpaper,11pt]{article}
\usepackage[margin=1in]{geometry}
\usepackage{amsthm}
\usepackage{amsfonts}
\usepackage{amssymb}
\usepackage{amsmath}
\usepackage[linesnumbered,algoruled,noend]{algorithm2e}
\SetNlSty{textbf}{}{.}
\usepackage[newenum]{paralist}
\usepackage[usenames,dvipsnames]{xcolor}
\usepackage[algoruled,noend]{algorithm2e}

\usepackage{amsthm}
\newtheorem{theorem}{Theorem}[section]
\newtheorem{lemma}[theorem]{Lemma}

\newtheorem{corollary}[theorem]{Corollary}
\newtheorem{problem}[theorem]{Problem}

\newcommand{\namedref}[2]{\hyperref[#2]{#1~\ref*{#2}}}
\newcommand{\sectionref}[1]{\namedref{Section}{#1}}

\newcommand{\theoremref}[1]{\namedref{Theorem}{#1}}

\newcommand{\lemmaref}[1]{\namedref{Lemma}{#1}}

\newcommand{\corollaryref}[1]{\namedref{Corollary}{#1}}
\newcommand{\probref}[1]{\namedref{Problem}{#1}}
\newcommand{\algref}[1]{\namedref{Algorithm}{#1}}

\newcommand{\stepref}[1]{\namedref{Step}{#1}}
\newcommand{\equalityref}[1]{\hyperref[#1]{Equality~\eqref{#1}}}
\newcommand{\inequalityref}[1]{\hyperref[#1]{Inequality~\eqref{#1}}}

\newcommand{\BO}{\mathcal{O}}
\newcommand{\N}{\mathbb{N}}

\hypersetup{colorlinks,linkcolor=blue,filecolor=blue,citecolor=blue,urlcolor=blue,pdfstartview=FitH}

\begin{document}
\title{Optimal Deterministic Routing and Sorting\\
on the Congested Clique}
\author
{Christoph Lenzen \\ Massachusetts Institute of Technology \\ Cambridge,
MA, USA \\ clenzen@csail.mit.edu}
\date{}
\maketitle

\begin{abstract}
Consider a clique of $n$ nodes, where in each synchronous round each pair of
nodes can exchange $\BO(\log n)$ bits. We provide deterministic constant-time
solutions for two problems in this model. The first is a routing problem where
each node is source and destination of $n$ messages of size $\BO(\log n)$. The
second is a sorting problem where each node $i$ is given $n$ keys of size
$\BO(\log n)$ and needs to receive the $i^{th}$ batch of $n$ keys according to
the global order of the keys. The latter result also implies deterministic
constant-round solutions for related problems such as selection or determining
modes.
\end{abstract}
% \setcounter{page}{0}
% \thispagestyle{empty}
% \newpage
\section{Introduction \& Related Work}\label{sec:intro}

Arguably, one of the most fundamental questions in distributed computing is what
amount of communication is required to solve a given task. For systems where
communication is dominating the ``cost''---be it the time to communicate
information, the money to purchase or rent the required infrastructure, or
any other measure derived from a notion of communication
complexity---exploring the imposed limitations may lead to more efficient
solutions.

Clearly, in such systems it does not make sense to make the complete input
available to all nodes, as this would be too expensive; typically, the same is
true for the output. For this reason, one assumes that each node is given a part
of the input, and each node needs to compute a corresponding part of the output.
For graph theoretic questions, the local input comprises the neighborhood of the
node in the respective graph, potentially augmented by weights for its incident
edges or similar information that is part of the problem specification. The
local output then e.g.\ consists of indication of membership in a set forming
the global solution (a dominating set, independent set, vertex cover, etc.),
a value between $0$ and $1$ (for the fractional versions), a color, etc. For
verification problems, one is satisfied if for a valid solution all nodes output
``yes'' and at least one node outputs ``no'' for an invalid solution.

Since the advent of distributed computing, a main research focus has been the
\emph{locality} of such computational problems. Obviously, one cannot compute,
or even verify, a spanning tree in less than $D$ synchronous communication
rounds, where $D$ is the diameter of the graph, as it is impossible to ensure
that a subgraph is acyclic without knowing it completely. Formally, the
respective lower bound argues that there are instances for which no node can
reliably distinguish between a tree and a non-tree since only the local
graph topology (and the parts of the prospective solution) up to distance $R$
can affect the information available to a node after $R$ rounds. More subtle
such \emph{indistinguishability} results apply to problems that \emph{can} be
solved in $o(D)$ time (see e.g.~\cite{kuhn10,linial92,naor91}).

This type of argument breaks down in systems where all nodes can communicate
directly or within a few number of rounds. However, this does not necessitate
the existence of efficient solutions, as due to limited bandwidth usually one
has to be selective in what information to actually communicate. This renders
otherwise trivial tasks much harder, giving rise to strong lower bounds. For
instance, there are $n$-node graphs of constant diameter on which finding or
verifying a spanning tree and many related problems require
$\tilde{\Omega}(\sqrt{n})$ rounds if messages contain a number of bits that is
polylogarithmic in $n$~\cite{elkin06,peleg00near,dasSarma2011}; approximating
the diameter up to factor $3/2-\varepsilon$ or determining it exactly cannot be
done in $\tilde{o}(\sqrt{n})$ and $\tilde{o}(n)$ rounds,
respectively~\cite{frischknecht12}. These and similar lower bounds consider
specific graphs whose topology prohibits to communicate efficiently. While the
diameters of these graphs are low, necessitating a certain connectivity, the
edges ensuring this property are few. Hence, it is impossible to transmit a
linear amount of bits between some nodes of the graph quickly, which forms the
basis of the above impossibility results.

This poses the question whether non-trivial lower bounds also hold in the case
where the communication graph is well-connected. After all, there are many
networks that do not feature small cuts, some due to natural expansion
properties, others by design. Also, e.g.\ in overlay networks, the underlying
network structure might be hidden entirely and algorithms may effectively
operate in a fully connected system, albeit facing bandwidth limitations.
Furthermore, while for scalability reasons full connectivity may not be
applicable on a system-wide level, it could prove useful to connect multiple
cliques that are not too large by a sparser high-level topology.

These considerations motivate to study distributed algorithms for a fully
connected system of $n$ nodes subject to a bandwidth limitation of $\BO(\log n)$
bits per round and edge, which is the topic of the present paper. Note that such
a system is very powerful in terms of communication, as each node can send and
receive $\Theta(n \log n)$ bits in each round, summing up to a total of
$\Theta(n^2\log n)$ bits per round. Consequently, it is not too surprising that,
to the best of our knowledge, so far no negative results for this model have
been published. On the positive side, a minimum spanning tree can be constructed
in $\BO(\log \log n)$ rounds~\cite{lotker06}, and, given to each node the
neighbors of a corresponding node in some graph as input, it can be decided
within $\BO(n^{1/3}/\log n)$ rounds whether the input graph contains a
triangle~\cite{lenzen12}. These bounds are deterministic;
constant-round randomized algorithms have been devised for the
routing~\cite{lenzen11} and sorting~\cite{patt-shamir11} tasks that we solve
deterministically in this work. The randomized solutions are about $2$ times as
fast, but there is no indication that the best deterministic algorithms are
slower than the best randomized algorithms.

\subsection*{Contribution}

We show that the following closely related problems
can be deterministically solved, within a constant number of communication
rounds in a fully connected system where messages are of size $\BO(\log n)$.
\begin{compactenum}
  \item [\textbf{Routing:}] Each node is source and destination of (up to) $n$
  messages of size $\BO(\log n)$. Initially only the sources know destinations
  and contents of their messages. Each node needs to learn all
  messages it is the destination of. (\sectionref{sec:routing})
  \item [\textbf{Sorting:}] Each node is given (up to) $n$ comparable keys of
  size $\BO(\log n)$. Node $i$ needs to learn about the keys with indices
  $(i-1)n+1,\ldots,in$ in a global enumeration of the keys that respects their
  order. Alternatively, we can require that nodes need to learn the indices of
  their keys in the total order of the union of all keys (i.e., all duplicate
  keys get the same index). Note that this implies constant-round solutions for
  related problems like selection or determining modes. (\sectionref{sec:sort})
\end{compactenum}
We note that the randomized algorithms from previous work are structurally very
different from the presented deterministic solutions. They rely on
near-uniformity of load distributions obtained by choosing intermediate
destinations uniformly and independently at random, in order to achieve
bandwidth-efficient communication. In contrast, the presented approach achieves
this in a style that has the flavor of a recursive sorting algorithm (with
a single level of recursion).

While our results are no lower bounds for well-connected systems under the
CONGEST model, they shed some light on why it is hard to prove impossibilities
in this setting: Even without randomization, the overhead required for
coordinating the efforts of the nodes is constant. In particular, any potential
lower bound for the considered model must, up to constant factors, also apply in
a system where each node can in each round send and receive $\Theta(n\log n)$
bits to and from arbitrary nodes in the system, with no further constraints on
communication.

We note that due to this observation, our results on sorting can equally well be
followed as corollaries of our routing result and Goodrich's sorting algorithm
for a bulk-synchronous model~\cite{goodrich99}. However, the derived algorithm
is more involved and requires at least an order of magnitude more rounds.

Since for such fundamental tasks as routing and sorting the amount of local
computations and memory may be of concern, we show in
\sectionref{sec:computations} how our algorithms can be adapted to require
$\BO(n\log n)$ computational steps and memory bits per node.
Trivially, these bounds are near-optimal with respect to computations and
optimal with respect to memory (if the size of the messages that are to be
exchanged between the nodes is $\Theta(\log n)$).
To complete the picture, in \sectionref{sec:size} we vary the parameters of
bandwidth, message/key size, and number of messages/keys per node. Our
techniques are sufficient to obtain asymptotically optimal results for almost
the entire range of parameters. For keys of size $o(\log n)$, we show that in
fact a huge number of keys can be sorted quickly; this is the special case for
which our bounds might not be asymptotically tight.

\section{Model}\label{sec:model}

In brief, we assume a fully connected system of $n$ nodes under the congestion
model. The nodes have unique identifiers $1$ to $n$ that are known to all other
nodes. Computation proceeds in synchronous rounds, where in each round, each
node performs arbitrary, finite computations,\footnote{Our algorithms will
perform polynomial computations with small exponent only.} sends a message to
each other node, and receives the messages sent by other nodes. Messages are of
size $\BO(\log n)$, i.e., in each message nodes may encode a constant number of
integer numbers that are polynomially bounded in $n$.\footnote{We will not
discuss this constraint when presenting our algorithms and only reason in a few
places why messages are not too large; mostly, this should be obvious from the
context.} To simplify the presentation, nodes will treat also themselves as
receivers, i.e., node $i\in \{1,\ldots,n\}$ will send messages to itself like to
any other node $j\neq i$.

These model assumptions correspond to the congestion model on the complete graph
$K_n=(V,\binom{V}{2})$ on the node set $V=\{1,\ldots,n\}$ (cf.~\cite{peleg00}).
We stress that in a given round, a node may send different messages along each
of its edges and thus can convey a total of $\Theta(n \log n)$ bits of
information. As our results demonstrate, this makes the considered model much
stronger than one where in any given round a node must broadcast the same
$\Theta(\log n)$ bits to all other nodes.

When measuring the complexity of computations performed by the nodes, we assume
that basic arithmetic operations on $\BO(\log n)$-sized values are a single
computational step.

\section{Routing}\label{sec:routing}

In this section, we derive a deterministic solution to the following task
introduced in~\cite{lenzen11}.
\begin{problem}[Information Distribution Task]\label{prob:idt}\ \\
Each node $i\in V$ is given a set of $n$ messages of size $\BO(\log n)$
\begin{equation*}
{\cal S}_i=\{m_i^1,\ldots,m_i^n\}
\end{equation*}
with destinations $d(m_i^j)\in V$, $j\in \{1,\ldots,n\}$. Messages are globally
lexicographically ordered by their source $i$, their destination $d(m_i^j)$, and
$j$. For simplicity, each such message explicitly contains these values, in
particular making them distinguishable. The goal is to deliver all messages to
their destinations, minimizing the total number of rounds. By
\begin{equation*}
{\cal R}_k:=\left\{m_i^j\in \bigcup_{i\in V}{\cal
S}_i\,\Bigg|\,d(m_i^j)=k\right\}
\end{equation*}
we denote the set of messages a node $k\in V$ shall receive. We require that
$|{\cal R}_k|= n$ for all $k\in V$, i.e., also the number of messages a
single node needs to receive is $n$.
\end{problem}
We remark that it is trivial to relax the requirement that each node needs to
send and receive \emph{exactly} $n$ messages; this assumption is made to
simplify the presentation. If each node sends/receives at most $n$ messages, our
techniques can be applied without change, and instances with more than $n$
sent/received messages per node can be split up into smaller ones.

\subsection{Basic Communication Primitives}
Let us first establish some basic communication patterns our algorithms will
employ. We will utilize the following classical result.

\begin{theorem}[Koenig's Line Coloring Theorem]\label{theorem:koenig}\ \\
Every $d$-regular bipartite multigraph is a disjoint union of $d$ perfect
matchings.
\end{theorem}
\begin{proof}
See e.g.\ Theorem 1.4.18 in \cite{lovasz09}.
\end{proof}
We remark that such an optimal coloring can be computed
efficiently~\cite{cole01}.\footnote{Also, a simple greedy coloring of the line
graph results in at most $2d-1$ (imperfect) matchings, which is sufficient for our
purposes. This will be used in \sectionref{sec:computations} to reduce the
amount of computations performed by the algorithm.}

Using this theorem, we can solve \probref{prob:idt} efficiently provided that it
is known a priori to all nodes what the sources and destinations of messages
are, an observation already made in~\cite{lenzen12}. We will however need a more
general statement applying to subsets of nodes that want to communicate among
themselves. To this end, we first formulate a generalization of the
result from~\cite{lenzen12}.

\begin{corollary}\label{coro:2_round}
We are given a subset $W\subseteq V$ and a bulk of messages such that the
following holds.
\begin{compactenum}
\item The source and destination of each message is in $W$.
\item The source and destination of each message is known in advance to all
nodes in $W$, and each source knows the contents of the messages to send.
\item Each node is the source of $f|W|$ messages, where $f:=\lfloor
n/|W|\rfloor$.
\item Each node is the destination of $f|W|$ messages.
\end{compactenum}
Then a routing scheme to deliver all messages in $2$ rounds can be found
efficiently. The routing scheme makes use of edges with at least one endpoint in
$W$ only.
\end{corollary}
\begin{proof}
Consider the bipartite multigraph $G=(S\dot{\cup}R,E)$ with $|S|=|R|=|W|$, where
$S=\{1_s,\ldots,|W|_s\}$ and $R=\{1_r,\ldots,|W|_r\}$ represent the nodes in
their roles as senders and receivers, respectively, and each input message at
some node $i$ that is destined for some node $j$ induces an edge from $i_s$
to~$j_r$.

By \theoremref{theorem:koenig}, we can color the edge set of $G$ with
$m:=f|W|\leq n$ colors such that no two edges with the same color have a node in
common. Moreover, as all nodes are aware of the source and destination of each
message, they can deterministically and locally compute the same such coloring,
without the need to communicate. Now, in the first communication round, each
node sends its (unique) message of color $c\in \{1,\ldots,m\}$ to node $c$. As
each node holds exactly one message of each color, at most one message is sent
over each edge, i.e., by the assumptions of the corollary this step can indeed
be performed in one round. Observe that this rule ensures that each node will
receive exactly one message of each color in the first round. Hence, because the
coloring also guarantees that each node is the destination of exactly one
message of each color, it follows for each $i,j\in \{1,\ldots,n\}$ that node $i$
receives exactly $f$ messages that need to be delivered to node $j$ in the first
round. Therefore all messages can be delivered by directly sending them to their
destinations in the second round.
\end{proof}

We stress that we can apply this result concurrently to multiple disjoint
sets $W$, provided that each of them satisfies the prerequisites of the
corollary: since in each routing step, each edge has at least one endpoint in
$W$, there will never be an edge which needs to convey more than one message in
each direction. This is vital for the success of our algorithms.

An observation that will prove crucial for our further reasoning is that for
subsets of size at most $\sqrt{n}$, the amount of information that needs to be
exchanged in order to establish common knowledge on the sources and destinations
of messages becomes sufficiently small to be handled. Since this information
itself consists, for each node, of $|W|$ numbers that need to be communicated to
$|W|\leq n/|W|$ nodes---with sources and destination known a priori!---we can
solve the problem for \emph{unknown} sources and destinations by applying the
previous corollary twice.
\begin{corollary}\label{coro:4_round}
We are given a subset $W\subseteq V$, where $|W|\leq \sqrt{n}$, and a bulk of
messages such that the following holds.
\begin{compactenum}
\item The source and destination of each message is in $W$.
\item Each source knows the contents of the messages to send.
\item Each node is the source of $f|W|$ messages, where $f:=\lfloor
n/|W|\rfloor$.
\item Each node is the destination of $f|W|$ messages.
\end{compactenum}
Then a routing scheme to deliver all messages in $4$ rounds can be found
efficiently. The routing scheme makes use of edges with at least one endpoint in
$W$ only.
\end{corollary}
\begin{proof}
Each node in $W$ announces the number of messages it holds for each node in $W$
to all nodes in $W$. This requires each node in $W$ to send and receive
$|W|^2\leq f|W|$ messages. As sources and destinations of these helper messages
are known in advance, by \corollaryref{coro:2_round} we can perform this
preprocessing in $2$ rounds. The information received establishes the
preconditions of \corollaryref{coro:2_round} for the original set of messages,
therefore the nodes now can deliver all messages in another two rounds.
\end{proof}

\subsection{Solving the Information Distribution Task}
Equipped with the results from the previous section, we are ready to tackle
\probref{prob:idt}. In the pseudocode of our algorithms, we will use a number of
conventions to allow for a straightforward presentation. When we state that a
message is \emph{moved} to another node, this means that the receiving node will
store a copy and serve as the source of the message in subsequent rounds of the
algorithm, whereas the original source may ``forget'' about the message. A step
where messages are moved is thus an actual routing step of the algorithm; all
other steps serve to prepare the routing steps. The current source of a message
\emph{holds} it. Moreover, we will partition the node set into subsets of size
$\sqrt{n}$, where for simplicity we assume that $\sqrt{n}$ is integer. We will
discuss the general case in the main theorem. We will frequently refer to these
subsets, where $W$ will invariably denote any of the sets in its role as source,
while $W'$ will denote any of the sets in its role as receiver (both with
respect to the current step of the algorithm). Finally, we stress that
statements about moving and sending of messages in the pseudocode do not imply
that the algorithm does so by direct communication between sending and receiving
nodes. Instead, we will discuss fast solutions to the respective (much simpler)
routing problems in our proofs establishing that the described strategies can be
implemented with small running times.

This being said, let us turn our attention to \probref{prob:idt}. The high-level
strategy of our solution is given in \algref{algo:high-level}.
\begin{algorithm}[ht]
\caption{High-level strategy for solving
\probref{prob:idt}.}\label{algo:high-level}
Partition the nodes into the disjoint subsets
  $\{(i-1)\sqrt{n}+1,\ldots,i\sqrt{n}\}$ for $i\in \{1,\ldots,\sqrt{n}\}$.\\
Move the messages such that each such subset $W$ holds exactly $|W||W'|=n$
  messages for each subset $W'$.\nllabel{line:high_2}\\
For each pair of subsets $W$, $W'$, move all messages destined to nodes
  in $W'$ within $W$ such that each node in $W$ holds exactly $|W'|=\sqrt{n}$
  messages with destinations in $W'$.\nllabel{line:high_3}\\
For each pair of subsets $W$, $W'$, move all messages destined to nodes
  in $W'$ from $W$ to $W'$.\nllabel{line:high_4}\\
For each $W$, move all messages within $W$ to their
destinations.\nllabel{line:high_5}\\
\end{algorithm}

Clearly, following this strategy will deliver all messages to their
destinations. In order to prove that it can be deterministically executed in a
constant number of rounds, we now show that all individual steps can be
performed in a constant number of rounds. Obviously, the first step requires no
communication. We leave aside \stepref{line:high_2} for now and turn to
\stepref{line:high_3}.
\begin{corollary}\label{coro:step_3}
\stepref{line:high_3} of \algref{algo:high-level} can be implemented in $4$
rounds.
\end{corollary}
\begin{proof}
The proof is analogous to \corollaryref{coro:4_round}. First, each node in $W$
announces to each other node in $W$ the number of messages it holds for each set
$W'$. By \corollaryref{coro:2_round}, this step can be completed in $2$ rounds,
for all sets $W$ in parallel.

With this information, the nodes in $W$ can deterministically compute
(intermediate) destinations for each message in $W$ such that the resulting
distribution of messages meets the condition imposed by \stepref{line:high_3}.
Applying \corollaryref{coro:2_round} once more, this redistribution can be
performed in another $2$ rounds, again for all sets $W$ concurrently.
\end{proof}

Trivially, \stepref{line:high_4} can be executed in a single round by each
node in $W$ sending exactly one of the messages with destination in $W'$ it
holds to each node in $W'$. According to \corollaryref{coro:4_round},
\stepref{line:high_5} can be performed in $4$ rounds.

Regarding \stepref{line:high_2}, we follow similar ideas. \algref{algo:step_2}
breaks our approach to this step down into smaller pieces.
\begin{algorithm}[ht]
\caption{\stepref{line:high_2} of \algref{algo:high-level} in more
detail.}\label{algo:step_2}
Each subset $W$ computes, for each set $W'$, the number of messages its
  constituents hold in total for nodes in $W'$. The results are announced to all
  nodes.\nllabel{line:step_2_1}\\
All nodes locally compute a pattern according to which the messages are to be
  moved between the sets. It satisfies that from each set $W$ to each set $W'$,
  $n$ messages need to be sent, and that in the resulting configuration, each
  subset $W$ holds exactly $|W||W'|=n$ messages for each subset
  $W'$.\nllabel{line:step_2_2}\\
All nodes in subset $W$ announce to all other nodes in $W$ the number of
  messages they need to move to each set $W'$ according
  to the previous step.\nllabel{line:step_2_3}\\
All nodes in $W$ compute a pattern for moving messages within $W$ so that the
  resulting distribution permits to realize the exchange computed in Step~2 in a
  single round (i.e., each node in $W$ must hold exactly $|W'|=\sqrt{n}$
  messages with (intermediate) destinations in $W'$).\nllabel{line:step_2_4}\\
The redistribution within the sets according to Step~4 is
executed.\nllabel{line:step_2_5}\\
The redistribution among the sets computed in Step~2 is
executed.\nllabel{line:step_2_6}
\end{algorithm}

We now show that following the sequence given in \algref{algo:step_2},
\stepref{line:high_2} of \algref{algo:high-level} requires a constant number of
communication rounds only.
\begin{lemma}\label{lemma:step_2}
\stepref{line:high_2} of \algref{algo:high-level} can be implemented in $7$
rounds.
\end{lemma}
\begin{proof}
We will show for each of the six steps of \algref{algo:step_2} that it can be
performed in a constant number of rounds and that the information available to
the nodes is sufficient to deterministically compute message exchange patterns
the involved nodes agree upon.

Clearly, \stepref{line:step_2_1} can be executed in two rounds. Each node in
$W$ simply sends the number of messages with destinations in the $i^{th}$ set
$W'$ it holds, where $i\in \{1,\ldots,\sqrt{n}\}$, to the $i^{th}$ node in $W$.
The $i^{th}$ node in $W$ sums up the received values and announces the result to
all nodes.

Regarding \stepref{line:step_2_2}, consider the following bipartite graph
$G=(S\dot{\cup} R,E)$. The sets $S$ and $R$ are of size $\sqrt{n}$ and represent
the subsets $W$ in their role as senders and receivers, respectively. For each
message held by a node in the $i^{th}$ set $W$ with destination in the $j^{th}$
set $W'$, we add an edge from $i\in S$ to $j\in R$. Note that after
\stepref{line:step_2_1}, each node can locally construct this graph. As each
node needs to send and receive $n$ messages, $G$ is of uniform degree $n^{3/2}$.
By \theoremref{theorem:koenig}, we can color the edge set of $G$ with $n^{3/2}$
colors so that no two edges of the same color share a node. We require that a
message of color $c\in \{1,\ldots,n^{3/2}\}$ is sent to the $(c
\operatorname{mod} \sqrt{n})^{th}$ set. Hence, the requirement that exactly $n$
messages need to be sent from any set $W$ to any set $W'$ is met. By requiring
that each node uses the same deterministic algorithm to color the edge set of
$G$, we make sure that the exchange patterns computed by the nodes agree.

Note that a subtlety here is that nodes cannot yet determine the precise color
of the messages they hold, as they do not know the numbers of messages to sets
$W'$ held by other nodes in $W$ and therefore also not the index of their
messages according to the global order of the messages. However, each node has
sufficient knowledge to compute the number of messages it holds with destination
in set $W'$ (for each $W'$), as this number is determined by the total numbers
of messages that need to be exchanged between each pair $W$ and $W'$ and the
node index only. This permits to perform \stepref{line:step_2_3} and then
complete \stepref{line:step_2_2} based on the received
information.\footnote{Formally, this can be seen as a deferred completition of
\stepref{line:step_2_2}.}

As observed before, \stepref{line:step_2_3} can be executed quickly: Each node
in $W$ needs to announce $\sqrt{n}$ numbers to all other nodes in $W$, which by
\corollaryref{coro:2_round} can be done in $2$ rounds. Now the nodes are capable
of computing the color of each of their messages according to the assignment
from \stepref{line:step_2_2}.

With the information gathered in \stepref{line:step_2_3}, it is now feasible to
perform \stepref{line:step_2_4}. This can be seen by applying
\theoremref{theorem:koenig} again, for each set $W$ to the bipartite multigraph
$G=(W\dot{\cup}R,E)$, where $R$ represents the $\sqrt{n}$ subsets $W'$ in their
receiving role with respect to the pattern computed in \stepref{line:step_2_2},
and each edge corresponds to a message held by a node in $W$ with destination in
some $W'$. The nodes can locally compute this graph due to the information they
received in Steps~\ref{line:step_2_2} and~\ref{line:step_2_3}. As $G$ has degree
$n$, we obtain an edge-coloring with $n$ colors. Each node in $W$ will move a
message of color $i\in \{1,\ldots,n\}$ to the $(i \operatorname{mod}
\sqrt{n})^{th}$ node in $W$, implying that each node will receive for each $W'$
exactly $\sqrt{n}$ messages with destination in $W'$.

Since the exchange pattern computed in \stepref{line:step_2_4} is, for each $W$,
known to all nodes in $W$, by \corollaryref{coro:2_round} we can perform
\stepref{line:step_2_5} for all sets in parallel in $2$ rounds. Finally,
\stepref{line:step_2_6} requires a single round only, since we achieved that
each node holds for each $W'$ exactly $\sqrt{n}$ messages with destination in
$W'$ (according to the pattern computed in \stepref{line:step_2_2}), and thus
can send exactly one of them to each of the nodes in $W'$ directly.

Summing up the number of rounds required for each of the steps, we see that
$2+0+2+0+2+1=7$ rounds are required in total, completing the proof.
\end{proof}

Overall, we have shown that each step of \algref{algo:high-level} can be
executed in a constant number of rounds if $\sqrt{n}$ is integer. It is not hard
to generalize this result to arbitrary values of $n$ without incurring larger
running times.
\begin{theorem}\label{theorem:idt}
\probref{prob:idt} can be solved deterministically within $16$ rounds.
\end{theorem}
\begin{proof}
If $\sqrt{n}$ is integer, the result immediately follows from
\lemmaref{lemma:step_2}, \corollaryref{coro:step_3}, and
\corollaryref{coro:4_round}, taking into account that the fourth step of the
high-level strategy requires one round.

If $\sqrt{n}$ is not integer, consider the following three sets of nodes:
\begin{eqnarray*}
V_1&:=&\{1,\ldots,\lfloor \sqrt{n}\rfloor^2\},\\
V_2&:=&\{n-\lfloor \sqrt{n}\rfloor^2+1,\ldots,n\}, \text{ and}\\
V_3&:=&\{1,\ldots,n-\lfloor
\sqrt{n}\rfloor^2\}\cup \{\lfloor \sqrt{n}\rfloor^2+1,\ldots,n\}.
\end{eqnarray*}
$V_1$ and
$V_2$ satisfy that $|V_1|=|V_2|=\lfloor \sqrt{n}\rfloor^2$. Hence, we can apply
the result for an integer root to the subsets of messages for which either both
sender and receiver are in $V_1$ or, symmetrically, in $V_2$. Doing so in
parallel will increase the message size by a factor of at most $2$. Note that
for messages where sender and receiver are in $V_1\cap V_2$ we can simply delete
them from the input of one of the two instances of the algorithm that run
concurrently, and adding empty ``dummy'' messages, we see that it is irrelevant
that nodes may send or receive less than $n$ messages in the individual
instances.

Regarding $V_3$, denote for each node $i\in V_3$ by $S_i\subseteq {\cal S}_i$
the subset of messages for which $i$ and the respective receiver are neither
both in $V_1$ nor both in $V_2$. In other words, for each message in $S_i$
either $i\in V_1\cap V_3$ and the receiver is in $V_2\cap V_3$ or vice versa.
Each node $i\in V_3$ moves the $j^{th}$ message in $S_i$ to node $j$ (one
round). No node will receive more than $|V_2\cap V_3|=|V_1\cap V_3|$ messages
with destinations in $V_1\cap V_3$, as there are no more than this number of
nodes sending such messages. Likewise, at most $|V_2\cap V_3|$ messages for
nodes in $V_2\cap V_3$ are received. Hence, in the subsequent round, all nodes
can move the messages they received for nodes in $V_1\cap V_3$ to nodes in
$V_1\cap V_3$, and the ones received for nodes in $V_2\cap V_3$ to nodes in
$V_2\cap V_3$ (one round). Finally, we apply \corollaryref{coro:4_round} to each
of the two sets to see that the messages $\bigcup_{i\in V_3}S_i$ can be
delivered within $4$ rounds. Overall, this procedure requires $6$ rounds, and
running it in parallel with the two instances dealing with other messages will
not increase message size beyond $\BO(\log n)$. The statement of the theorem
follows.
\end{proof}

\section{Sorting}\label{sec:sort}

In this section, we present a deterministic sorting algorithm. The problem
formulation is essentially equivalent to the one in~\cite{patt-shamir11}.
\begin{problem}[Sorting]\label{prob:sort}
Each node is given $n$ keys of size $\BO(\log n)$ (i.e., a key fits into a
message). We assume w.l.o.g.\ that all keys are different.\footnote{Otherwise
we order the keys lexicographically by key, node whose input contains the key,
and a local enumeration of identical keys at each node.} Node $i$ needs to
learn the keys with indices $i(n-1)+1,\ldots,in$ according the total order of
all keys.
\end{problem}

\subsection{Sorting Fewer Keys with Fewer Nodes}
Again, we assume for simplicity that $\sqrt{n}$ is integer and deal with the
general case later on. Our algorithm will utilize a subroutine that can sort up
to $2n^{3/2}$ keys within a subset $W\subset V$ of $\sqrt{n}$ nodes,
communicating along edges with at least one endpoint in the respective subset of
nodes. The latter condition ensures that we can run the routine in parallel for
disjoint subsets $W$. We assume that each of the nodes in $W$ initially holds
$2n$ keys. The pseudocode of our approach is given in \algref{algo:sub_sort}.
\begin{algorithm}[ht]
\caption{Sorting $2n^{3/2}$ keys with $|W|=\sqrt{n}$ nodes. Each node in $W$
has $2n$ input keys and learns their indices in the total order of all
$2n^{3/2}$ keys.}\label{algo:sub_sort}
Each node in $W$ locally sorts its keys and selects every
  $(2\sqrt{n})^{th}$ key according to this order (i.e., a key of local index $i$
  is selected if $i \operatorname{mod} 2\sqrt{n} = 0$).\nllabel{line:sub_1}\\
Each node in $W$ announces the selected keys to all other nodes in $W$.
\nllabel{line:sub_2}\\
Each node in $W$ locally sorts the union of the received keys and
  selects every $\sqrt{n}^{th}$ key according to this order. We call such a key
  \emph{delimiter}.\nllabel{line:sub_3}\\
Each node $i\in W$ splits its original input into $\sqrt{n}$ subsets,
  where the $j^{th}$ subset $K_{i,j}$ contains all keys that are larger than the
  $(j-1)^{th}$ delimiter (for $j=1$ this condition does not apply) and smaller
  or equal to the $j^{th}$ delimiter.\nllabel{line:sub_4}\\
Each node $i\in W$ announces for each $j$ $|K_{i,j}|$ to all nodes in
  $W$.\nllabel{line:sub_5}\\
Each node $i\in W$ sends $K_{i,j}$ to the $j^{th}$ node in
  $W$.\nllabel{line:sub_6}\\
Each node in $W$ locally sorts the received keys. The sorted sequence now
consists of the concatenation of the sorted sequences in the order of the node
identifiers.\nllabel{line:sub_7}\\
Keys are redistributed such that each node receives $2n$
keys and the order is maintained.\nllabel{line:sub_8}
\end{algorithm}

Let us start out with the correctness of the proposed scheme.
\begin{lemma}\label{lemma:few_keys_correct}
When executing \algref{algo:sub_sort}, the nodes in $W$ are indeed capable of
computing their input keys' indices in the order on the union of the input keys
of the nodes in $W$.
\end{lemma}
\begin{proof}
Observe that because all nodes use the same input in \stepref{line:sub_3}, they
compute the same set of delimiters. The set of all keys is the union
$\bigcup_{j=1}^{\sqrt{n}}\bigcup_{i\in W}K_{i,j}$, and the sets $K_{i,j}$ are
disjoint. As the $K_{i,j}$ are defined by comparison with the delimiters, we
know that all keys in $K_{i,j}$ are larger than keys in $K_{i',j'}$ for all
$i'\in W$ and $j'<j$, and smaller than keys in $K_{i',j'}$ for all $i'\in W$ and
$j'>j$. Since in \stepref{line:sub_7} the received keys are locally sorted and
\stepref{line:sub_8} maintains the resulting order, correctness follows.
\end{proof}

Before turning to the running time of the algorithm, we show that the
partitioning of the keys by the delimiters is well-balanced.
\begin{lemma}\label{lemma:balanced}
When executing \algref{algo:sub_sort}, for each $j\in \{1,\ldots,\sqrt{n}\}$ it
holds that
\begin{equation*}
\left|\bigcup_{i\in W}K_{i,j}\right|< 4n.
\end{equation*}
\end{lemma}
\begin{proof}
Due to the choice of the delimiters, $\bigcup_{i\in W}K_{i,j}$ contains
exactly $\sqrt{n}$ of the keys selected in \stepref{line:sub_1} of the
algorithm. Denote by $d_i$ the number of such selected keys in $K_{i,j}$. As in
\stepref{line:sub_1} each node selects every $(2\sqrt{n})^{th}$ of its keys and
the set $K_{i,j}$ is a contiguous subset of the ordered sequence of input keys
at $w$, we have that $|K_{i,j}|<2\sqrt{n}(d_i+1)$. It follows that
\begin{equation*}
\left|\bigcup_{i\in W}K_{i,j}\right|=\sum_{i\in W}|K_{i,j}|
<2\sqrt{n}\sum_{i\in W}(d_i+1)=2\sqrt{n}(\sqrt{n}+|W|)=4n.\qedhere
\end{equation*}
\end{proof}

We are now in the position to complete our analysis of the subroutine.
\begin{lemma}\label{lemma:few_keys}
Given a subset $W\subseteq V$ of size $\sqrt{n}$ such that each $w\in W$ holds
$2n$ keys, each node in $W$ can learn about the indices of its keys in the total
order of all keys held by nodes in $W$ within $10$ rounds. Furthermore, only
edges with at least one endpoint in $W$ are used for this purpose.
\end{lemma}
\begin{proof}
By \lemmaref{lemma:few_keys_correct}, \algref{algo:sub_sort} is correct. Hence,
it remains to show that it can be implemented with $10$ rounds of communication,
using no edges with both endpoints outside $W$.

Steps~\ref{line:sub_1}, \ref{line:sub_3}, \ref{line:sub_4}, and~\ref{line:sub_7}
involve local computations only. Since $|W|=\sqrt{n}$ and each node selects
exactly $\sqrt{n}$ keys it needs to announce to all other nodes, according to
\corollaryref{coro:2_round} \stepref{line:sub_2} can be performed in $2$ rounds.
The same holds true for \stepref{line:sub_5}, as again each node needs to
announce $|W|=\sqrt{n}$ values to each other node in $W$. In
\stepref{line:sub_6}, each node sends its $2n$ input keys and, by
\lemmaref{lemma:balanced}, receives at most $4n$ keys. By bundling a constant
number of keys in each message, nodes need to send and receive at most
$n=|W|\cdot n/|W|$ messages. Hence, \corollaryref{coro:4_round} states that this
step can be completed in $4$ rounds. Regarding \stepref{line:sub_8}, observe
that due to \stepref{line:sub_5} each node knows how many keys each other node
holds at the beginning of the step. Again bundling a constant number of keys
into each message, we thus can apply \corollaryref{coro:2_round} to complete
\stepref{line:sub_8} in $2$ rounds. In total, we thus require $0+2+0+0+2+4+2=10$
communication rounds.

As we invoked Corollaries~\ref{coro:2_round} and~\ref{coro:4_round} in order to
define the communication pattern, it immediately follows from the
corollaries that all communication is on edges with at least one endpoint
in~$W$.
\end{proof}

\subsection{Sorting All Keys}
With this subroutine at hand, we can move on to \probref{prob:sort}. Our
solution follows the same pattern as \algref{algo:sub_sort}, where the
subroutine in combination with \theoremref{theorem:idt} enables that sets of
size $\sqrt{n}$ can take over the function nodes had in \algref{algo:sub_sort}.
This increases the processing power by factor $\sqrt{n}$, which is sufficient to
deal with all $n^2$ keys. \algref{algo:sort} shows the high-level structure of
our solution.

\begin{algorithm}[ht]
\caption{Solving \probref{prob:sort}.}\label{algo:sort}
Each node locally sorts its input and selects every $\sqrt{n}^{th}$ key
  (i.e., the index in the local order modulo $\sqrt{n}$ equals
  $0$).\nllabel{line:sort_1}\\
Each node transmits its $i^{th}$ selected key to node
  $i$.\nllabel{line:sort_2}\\
Using \algref{algo:sub_sort}, nodes $1,\ldots,\sqrt{n}$ sort the in total
  $n^{3/2}$ keys they received (i.e., determine the respective indices in the
  induced order).\nllabel{line:sort_3}\\
Out of the sorted subsequence, every $n^{th}$ key is selected as
  \emph{delimiter} and announced to all nodes (i.e., there is a total of
  $\sqrt{n}$ delimiters).\nllabel{line:sort_4}\\
Each node $i\in V$ splits its original input into $\sqrt{n}$ subsets,
  where the $j^{th}$ subset $K_{i,j}$ contains all keys that are larger than the
  $(j-1)^{th}$ delimiter (for $j=1$ this condition does not apply) and smaller
  or equal to the $j^{th}$ delimiter.\nllabel{line:sort_5}\\
The nodes are partitioned into $\sqrt{n}$ disjoint sets $W$ of size
  $\sqrt{n}$. Each node $i\in V$ sends $K_{i,j}$ to the $j^{th}$ set $W$ (i.e.,
  each node in $W$ receives either $\lfloor|K_{i,j}|/|W|\rfloor$ or
  $\lceil|K_{i,j}|/|W|\rceil$ keys, and each key is sent to exactly one
  node).\nllabel{line:sort_6}\\
Using \algref{algo:sub_sort}, the sets $W$ sort the received keys.\nllabel{line:sort_7}\\
Keys are redistributed such that each node receives
$n$ keys and the order is maintained.\nllabel{line:sort_8}
\end{algorithm}

The techniques and results from the previous sections are sufficient to derive
our second main theorem without further delay.
\begin{theorem}\label{theorem:sort}
\probref{prob:sort} can be solved in $37$ rounds.
\end{theorem}
\begin{proof}
We discuss the special case of $\sqrt{n}\in \N$ first, to which we can apply
\algref{algo:sort}. Correctness of the algorithm follows analogously to
\lemmaref{lemma:few_keys_correct}. Steps~\ref{line:sort_1} and
\ref{line:sort_5} require local computations only. \stepref{line:sort_2}
involves one round of communication. \stepref{line:sort_3} calls
\algref{algo:sub_sort}, which by \lemmaref{lemma:few_keys} consumes $10$ rounds.
However, we can skip the last step of the algorithm and instead directly execute
\stepref{line:sort_4}. This takes merely $2$ rounds, since there are
$\sqrt{n}$ nodes each of which needs to announce at most $2\sqrt{n}$ values to
all nodes and we can bundle two values in one message. Regarding
\stepref{line:sort_6}, observe that, analogously to \lemmaref{lemma:balanced},
we have for each $j\in \{1,\ldots,\sqrt{n}\}$ that
\begin{equation*}
\left|\bigcup_{i\in V}K_{i,j}\right|=\sum_{i\in V}|K_{i,j}|<
\sqrt{n}(n+|V|)=2n^{3/2}.
\end{equation*}
Hence, each node needs to send at most $n$ keys and receive at most $2n$ keys.
Bundling up to two keys in each message, nodes need to send and receive at most
$n$ messages. Therefore, by \theoremref{theorem:idt}, \stepref{line:sort_6} can
be completed within $16$ rounds. \stepref{line:sub_7} again calls
\algref{algo:sub_sort}, this time in parallel for all sets $W$. Nonetheless, by
\lemmaref{lemma:few_keys} this requires $10$ rounds only because the edges used
for communication are disjoint. Also here, we can skip the last step of the
subroutine and directly move on to \stepref{line:sort_8}. Again,
\corollaryref{coro:2_round} implies that this step can be completed in $2$
rounds. Overall, the algorithm runs for $0+1+8+2+0+16+8+2=37$ rounds.

With respect to non-integer values of $\sqrt{n}$, observe that we can increase
message size by any constant factor to accommodate more keys in each message.
This way we can work with subsets of size $\lfloor \sqrt{n}\rfloor$ and
similarly select keys and delimiters in Steps~\ref{line:sort_1}
and~\ref{line:sort_4} such that the adapted algorithm can be completed in $37$
rounds as well.
\end{proof}

We conclude this section with a corollary stating that the slightly modified
task of determining each input key's position in a global enumeration of the
\emph{different} keys that are present in the system can also be solved
efficiently. Note that this implies constant-round solutions for determining
modes and selection as well.
\begin{corollary}\label{coro:sort}
Consider the variant of \probref{prob:sort} in which each node is required to
determine the index of its input keys in the total order of the union of all
input keys. This task can be solved deterministically in a constant number of
rounds.
\end{corollary}
\begin{proof}
After applying the sorting algorithm, each node announces (i) its smallest and
largest key, (ii) how many copies of each of these two keys it holds, and (iii)
the number of distinct keys it holds to all other nodes. This takes one round,
and from this information all nodes can compute the indices in the
non-repetitive sorted sequence for their keys. Applying
\theoremref{theorem:idt}, we can inform the nodes whose input the keys were of
these values in a constant number of rounds.
\end{proof}

\section{Local Computations and Memory Requirements}\label{sec:computations}

Examining Algorithms~\ref{algo:high-level} and~\ref{algo:step_2} and how we
implemented their various steps, it is not hard to see that all computations
that do not use the technique of constructing some bipartite multigraph and
coloring its edges merely require $\BO(n)$ computational steps (and thus, as
all values are of size $\BO(\log n)$, also $\BO(n \log n)$ memory). Leaving the
work and memory requirements of local sorting operations aside, the same applies
to Algorithms~\ref{algo:sub_sort} and~\ref{algo:sort}. Assuming that an
appropriate sorting algorithm is employed, the remaining question is how
efficiently we can implement the steps that do involve coloring.

The best known algorithm to color a bipartite multigraph $H=(V,E)$ of maximum
degree $\Delta$ with $\Delta$ colors requires $\BO(|E|\log \Delta)$
computational steps~\cite{cole01}. Ensuring that $|E|\in \BO(n)$ in all cases
where we appeal to the procedure will thus result in a complexity of $\BO(n\log
n)$. Unfortunately, this bound does not hold for the presented algorithms. More
precisely, \stepref{line:high_3} of \algref{algo:high-level} and
Steps~\ref{line:step_2_2} and~\ref{line:step_2_4} of \algref{algo:step_2}
violate this condition. Let us demonstrate first how this issue can be resolved
for \stepref{line:high_3} of \algref{algo:high-level}.
\begin{lemma}\label{lemma:computation_high_3}
Steps~\ref{line:high_3} and~\ref{line:high_4} of \algref{algo:high-level} can be
executed in $3$ rounds such that each node performs $\BO(n)$ steps of local
computation.
\end{lemma}
\begin{proof}
Each node locally orders the messages it holds according to their destination
sets $W'$; using bucketsort, this can be done using $\BO(n)$ computational
steps. According to this order, it moves its messages to the nodes in $W$
following a round-robin pattern. In order to achieve this in $2$ rounds, it
first sends to each other node in the system one of the messages; in the second
round, these nodes forward these messages to nodes in $W$. Since an appropriate
communication pattern can be fixed independently of the specific distribution of
messages, no extra computations are required.

Observe that in the resulting distribution of messages, no node in $W$ holds
more than $2\sqrt{n}$ messages for each set $W'$: For every full $\sqrt{n}$
messages some node in $W$ holds for set $W'$, every node in $W$ gets exactly one
message destined for $W'$, plus possible one residual message for each node in
$W$ that does not hold an integer multiple of $\sqrt{n}$ messages for $W'$.
Hence, moving at most two messages across each edge in a single round,
\stepref{line:high_4} can be completed in one round.
\end{proof}
Note that we save two rounds for \stepref{line:high_3} in comparison to
\corollaryref{coro:computation_step_2_4}, but at the expense of doubling the
message size in \stepref{line:high_4}.

The same argument applies to \stepref{line:step_2_4} of \algref{algo:step_2}.
\begin{corollary}\label{coro:computation_step_2_4}
Steps~\ref{line:step_2_3} to~\ref{line:step_2_5} of \algref{algo:step_2} can be
executed in $2$ rounds, where each node performs $\BO(n)$ steps of local
computation.
\end{corollary}

\stepref{line:step_2_2} of \algref{algo:step_2} requires a different approach
still relying on our coloring construction.
\begin{lemma}\label{lemma:computation_step_2}
A variant of \algref{algo:step_2} can execute \stepref{line:high_2} of
\algref{algo:high-level} in $5$ rounds using $\BO(n\log n)$ steps of local
computation and memory bits at each node.
\end{lemma}
\begin{proof}
As mentioned before, the critical issue is that Steps~\ref{line:step_2_2}
and~\ref{line:step_2_4} of \algref{algo:step_2} rely on bipartite graphs with
too many edges. \corollaryref{coro:computation_step_2_4} applies to
\stepref{line:step_2_4}, so we need to deal with \stepref{line:step_2_2} only.

To reduce the number edges in the graph, we group messages from $W$ to $W'$
into sets of size $n$. Note that not all respective numbers are integer
multiples of $n$, and we need to avoid ``incomplete'' sets of smaller
size as otherwise the number of edges still might be too large. This is
easily resolved by dealing with such ``residual'' messages by directly sending
them to their destinations: Each set will hold less than $n$ such messages for
each destination set $W'$ and therefore can deliver these messages using its $n$
edges to set $W'$.\footnote{The nodes account for such messages as well when
performing the redistribution of messages within $W$ in
Steps~\ref{line:step_2_3} to~\ref{line:step_2_5}.}

It follows that the considered bipartite multigraph will have $\BO(n)$ edges and
maximum degree $\sqrt{n}$. It remains to argue why all steps can be performed
with $\BO(n\log n)$ steps and memory at each node. This is obvious for
\stepref{line:step_2_1} and \stepref{line:step_2_6} and follows from
\corollaryref{coro:computation_step_2_4} for Steps~\ref{line:step_2_3}
to~\ref{line:step_2_5}. Regarding \stepref{line:step_2_2}, observe that the
bipartite graph considered can be constructed in $\BO(n)$ steps since this
requires adding $\sqrt{n}$ integers for each of the $\sqrt{n}$ destination sets
(and determining the integer parts of dividing the results by $n$). Applying the
algorithm from~\cite{cole01} then colors the edges within $\BO(n\log n)$ steps.
Regarding memory, observe that all other steps require $\BO(n)$ computational
steps and thus trivially satisfy the memory bound. The algorithm
from~\cite{cole01} computes the coloring by a recursive divide and conquer
strategy; clearly, an appropriate implementation thus will not require more than
$\BO(n\log n)$ memory either.
\end{proof}

We conclude that there is an implementation of our scheme that is
simultaneously efficient with respect to running time, message size, local
computations, and memory consumption.
\begin{theorem}
\probref{prob:idt} can be solved deterministically within $12$ rounds, where
each node performs $\BO(n\log n)$ steps of computation using $\BO(n\log n)$
memory bits.
\end{theorem}
This result immediately transfers to \probref{prob:sort}.
\begin{corollary}
\probref{prob:sort} and its variant discussed in \corollaryref{coro:sort} can
be solved in a constant number of rounds, where each node performs $\BO(n\log
n)$ steps of computation using $\BO(n\log n)$ memory bits.
\end{corollary}
\section{Varying Message and Key Size}\label{sec:size}

In this section, we discuss scenarios where the number and size of messages and
keys for Problems~\ref{prob:idt} and~\ref{prob:sort} vary. This also motivates
to reconsider the bound on the number bits that nodes can exchange in each
round: For message/key size of $\Theta(\log n)$, communicating $B\in \BO(\log
n)$ bits over each edge in each round was shown to be sufficient, and for
smaller $B$ the number of rounds clearly must increase
accordingly.\footnote{Formally proving a lower bound is trivial in both cases,
as nodes need to communicate their $n$ messages to deliver all messages or their
$n$ keys to enable determining the correct indices of all keys, respectively.}
We will see that most ranges for these parameters can be handled asymptotically
optimally by the presented techniques. For the remaining cases, we will give
solutions in this section. We remark that one can easily verify that the
techniques we propose in the sequel are also efficient with respect to local
computations and memory requirements.

\subsection{Large Messages or Keys}

If messages or keys contain $\omega(\log n)$ bits and $B$ is not sufficiently
large to communicate a single value in one message, splitting these values into
multiple messages is a viable option. For instance, with bandwidth $B\in
\Theta(\log n)$, a key of size $\Theta(\log^2 n)$ would be split into
$\Theta(\log n)$ separate messages permitting the receiver to reconstruct the
key from the individual messages. This simple argument shows that in fact not
the total number of messages (or keys) is decisive for the more general versions
of Problems~\ref{prob:idt} and~\ref{prob:sort}, but the number of bits that need
to be sent and received by each node. If this number is in $\Omega(n \log n)$,
the presented techniques are asymptotically optimal.

\subsection{Small Messages}
If we assume that in \probref{prob:sort} the size of messages is bounded by
$M\in o(\log n)$, we may hope that we can solve the problem in a constant number
of rounds even if we merely transmit $B\in \BO(M)$ bits along each edge.
With the additional assumption that nodes can identify the sender of a message
even if the identifier is not included, this can be achieved if sources and
destinations of messages are known in advance: We apply
\corollaryref{coro:2_round} and observe that because the communication pattern
is known to all nodes, knowing the sender of a message is sufficient to perform
the communication and infer the original source of each message at the
destination.

On the other hand, if sources/destinations are unknown, consider inputs where
$\Omega(n^2)$ messages cannot be sent directly from their sources to their
destinations (i.e., using the respective source-receiver edge) within a constant
number of rounds. Each of these messages needs to be forwarded in a way
preserving their destination, i.e., at least one of the forwarding nodes must
learn about the destination of the message (otherwise correct delivery cannot be
guaranteed). Explicitly encoding these values for $\Omega(n^2)$ messages
requires $\Omega(n^2\log n)$ bits. Implicit encoding can be done by means of the
round number or relations between the communication partners' identifiers.
However, encoding bits by introducing constraints reduces (at least for
worst-case inputs) the number of messages that can be sent by a node
accordingly. These considerations show that in case of \probref{prob:idt}, small
messages do not simplify the task.

\subsection{Small Keys}
The situation is different for \probref{prob:sort}. Note that we need to drop
the assumption that all keys can be distinguished, as this would necessitate key
size $\Omega(\log n)$. In contrast, if keys can be encoded with $o(\log n)$
bits, there are merely $n^{o(1)}$ different keys. Hence, we can statically
assign disjoint sets of $\log^2 n$ nodes to each key $\kappa$ (for simplicity we
assume that $\log n$ is integer). In the first round, each node binary encodes
the number of copies it holds of $\kappa$ and sends the $i^{th}$ bit to $\log n$
of these nodes. The $j^{th}$ of the $\log n$ receiving nodes of bit $i$ counts
the number of nodes which sent it a $1$, encodes this number binary, and
transmits the $j^{th}$ bit to all nodes. With this information, all nodes are
capable of computing the total number of copies of $\kappa$ in the system.

In order to assign an order to the different copies of $\kappa$ in the system
(if desired), in the second round we can require that in addition the $j^{th}$
node dealing with bit $i$ sends to node $k\in \{1,\ldots,n\}$ the $j^{th}$ bit
of an encoding of the number of nodes $k'\in \{1,\ldots,k-1\}$ that sent a $1$
in the first round. This way, node $k$ can also compute the number of copies of
$\kappa$ held by nodes $k'<k$, which is sufficient to order the keys as
intended.

It is noteworthy that this technique can actually be used to order a much larger
total number of keys, since we ``used'' very few of the nodes. If we have $K\leq
n/\log^2 n$ different keys, we can assign $m:=\lfloor n/K\rfloor$ nodes to each
key. This permits to handle any binary encoding of up to $\lfloor
\sqrt{m}\rfloor$ many bits in the above manner, potentially allowing for huge
numbers of keys. At the same time, messages contain merely $2$ bits (or a single
bit, if we accept $3$ rounds of communication). More generally, each node can be
concurrently responsible for $B$ bits, improving the power of the approach
further for non-constant values of $B$.

\section*{Acknowledgement}
The author would like to thank Shiri Chechic, Quentin Godfroy, Merav Parter, and
Jukka Suomela for valuable discussions. This material is based upon work
supported by the National Science Foundation under Grant Nos.\ CCF-AF-0937274,
CNS-1035199, 0939370-CCF and CCF-1217506,\linebreak the AFOSR under Contract
No.\ AFOSR Award number FA9550-13-1-0042, the Swiss National Science Foundation (SNSF), the
Swiss Society of Friends of the Weizmann Institute of Science, and the German
Research Foundation (DFG, reference number Le 3107/1-1).

\bibliographystyle{abbrv}
\bibliography{sorting}

\begin{thebibliography}{10}

\bibitem{cole01}
R.~Cole, K.~Ost, and S.~Schirra.
\newblock {Edge-Coloring Bipartite Multigraphs in $\BO(|E|\log \Delta)$ Time}.
\newblock {\em Combinatorica}, 21:5--12, 2001.

\bibitem{lenzen12}
D.~Dolev, C.~Lenzen, and S.~Peled.
\newblock {``Tri, Tri again'': Finding Triangles and Small Subgraphs in a
  Distributed Setting}.
\newblock In {\em Proc.\ 26th Symposium on Distributed Computing (DISC)}, pages
  195--209, 2012.

\bibitem{elkin06}
M.~Elkin.
\newblock {An Unconditional Lower Bound on the Time-Approximation Tradeoff for
  the Minimum Spanning Tree Problem}.
\newblock {\em SIAM Journal on Computing}, 36(2):463--501, 2006.

\bibitem{frischknecht12}
S.~Frischknecht, S.~Holzer, and R.~Wattenhofer.
\newblock {Networks Cannot Compute Their Diameter in Sublinear Time}.
\newblock In {\em Proc.\ 23rd Symposium on Discrete Algorithms (SODA)}, pages
  1150--1162, 2012.

\bibitem{goodrich99}
M.~T. Goodrich.
\newblock {Communication-Efficient Parallel Sorting}.
\newblock {\em SIAM Journal on Computing}, 29(2):416--432, Oct. 1999.

\bibitem{kuhn10}
F.~Kuhn, T.~Moscibroda, and R.~Wattenhofer.
\newblock {Local Computation: Lower and Upper Bounds}.
\newblock {\em Computing Research Repository}, abs/1011.5470, 2010.

\bibitem{lenzen11}
C.~Lenzen and R.~Wattenhofer.
\newblock {Tight Bounds for Parallel Randomized Load Balancing}.
\newblock In {\em Proc. 43rd Symposium on Theory of Computing (STOC)}, pages
  11--20, 2011.

\bibitem{linial92}
N.~Linial.
\newblock {Locality in Distributed Graph Algorithms}.
\newblock {\em SIAM Journal on Computing}, 21(1):193--201, 1992.

\bibitem{lotker06}
Z.~Lotker, B.~Patt-Shamir, and D.~Peleg.
\newblock {Distributed MST for Constant Diameter Graphs}.
\newblock {\em Distributed Computing}, 18(6):453--460, 2006.

\bibitem{lovasz09}
L.~Lov\'{a}sz and M.~D. Plummer.
\newblock {\em {Matching Theory}}.
\newblock American Mathematical Society, 2009.
\newblock Reprint with corrections.

\bibitem{naor91}
M.~Naor.
\newblock {A Lower Bound on Probabilistic Algorithms for Distributive Ring
  Coloring}.
\newblock {\em SIAM Journal on Discrete Mathematics}, 4(3):409--412, 1991.

\bibitem{patt-shamir11}
B.~Patt-Shamir and M.~Teplitsky.
\newblock {The Round Complexity of Distributed Sorting: Extended Abstract}.
\newblock In {\em Proc.\ 30th Symposium on Principles of Distributed Computing
  (PODC)}, pages 249--256, 2011.

\bibitem{peleg00}
D.~Peleg.
\newblock {\em {Distributed Computing: A Locality-Sensitive Approach}}.
\newblock {Society for Industrial and Applied Mathematics}, 2000.

\bibitem{peleg00near}
D.~Peleg and V.~Rubinovich.
\newblock {Near-tight Lower Bound on the Time Complexity of Distributed MST
  Construction}.
\newblock {\em SIAM Journal on Computing}, 30:1427--1442, 2000.

\bibitem{dasSarma2011}
A.~D. Sarma, S.~Holzer, L.~Kor, A.~Korman, D.~Nanongkai, G.~Pandurangan,
  D.~Peleg, and R.~Wattenhofer.
\newblock {Distributed Verification and Hardness of Distributed Approximation}.
\newblock {\em SIAM Journal on Computation}, 41(5):1235--1265, 2012.

\end{thebibliography}

\end{document}